\documentclass[11pt]{article}
\usepackage[margin=1in]{geometry}                		%
\newif\ifcomments
\commentstrue

\usepackage{fullpage}
\usepackage[utf8]{inputenc}

\usepackage{graphicx}
\usepackage{amsmath,amsthm,amssymb}
\usepackage{algorithm}
\usepackage{algpseudocode}
\usepackage{multirow}
\usepackage{makecell}

\usepackage{thm-restate,color,xspace}
\usepackage{comment}
\usepackage{thmtools}
\usepackage{xcolor}
\usepackage{nameref}
\usepackage{array}

\usepackage{tikz}
\usetikzlibrary{positioning}

\definecolor{ForestGreen}{rgb}{0.1333,0.5451,0.1333}
\definecolor{DarkRed}{rgb}{0.65,0,0}
\definecolor{Red}{rgb}{1,0,0}
\usepackage[linktocpage=true,
pagebackref=true,colorlinks,
linkcolor=DarkRed,citecolor=ForestGreen,
bookmarks,bookmarksopen,bookmarksnumbered]
{hyperref}
\usepackage{cleveref}
\usepackage{dsfont}
\usepackage[colorinlistoftodos,textsize=tiny,textwidth=2cm,color=red!25!white,obeyFinal]{todonotes}

\declaretheorem[numberwithin=section]{theorem}
\declaretheorem[numberlike=theorem]{lemma}

\declaretheorem[numberlike=theorem]{fact}

\declaretheorem[numberlike=theorem]{corollary}

\declaretheorem[numberlike=theorem]{claim}

\declaretheorem[numberlike=theorem,style=definition]{definition}

\global\long\def\polylog{\mathrm{polylog}}

\newcommand{\wtilde}{\widetilde}

\newcommand{\OO}{\mathcal{O}}
\newcommand{\val}{\textsc{val}}
\newcommand{\Cut}{\textsc{Cut}}

\newcommand{\Otil}{\widetilde{O}}

\newcommand{\vol}{\text{vol}}
\newcommand{\CC}{\mathcal{C}}
\global\long\def\poly{\mathrm{poly}}

\global\long\def\val{\mathrm{val}}

\global\long\def\dist{\mathrm{dist}}
\renewcommand{\arraystretch}{1.5}

\newcommand{\ignore}[1]{}

\def\DEBUG{truetrue}
\ifdefined\DEBUG
\newcommand{\anote}[1]{\todo[color=blue!25!white]{AA: #1}\xspace}

\def\thatchaphol#1{\marginpar{$\leftarrow$\fbox{T}}\footnote{$\Rightarrow$~{\sf\textcolor{red}{#1 --Thatchaphol}}}}
\else
\newcommand{\anote}[1]{}
\def\thatchaphol#1{}
\fi

\def\DEBUG{truetrue}
\ifdefined\DEBUG
\newcommand{\ynote}[1]{\todo[color=red!25!white]{YF: #1}\xspace}
\else
\newcommand{\ynote}[1]{}
\fi

\title{Deterministic Edge Connectivity and Max Flow\\ using Subquadratic Cut Queries}

\author{
    Aditya Anand \thanks{University of Michigan Ann Arbor, \texttt{adanand@umich.edu}.}\and 
    Thatchaphol Saranurak\thanks{
        University of Michigan Ann Arbor,
        \texttt{thsa@umich.edu}.
        Supported by NSF grant CCF-2238138. Partially funded by the Ministry of Education and Science of Bulgaria's support for INSAIT, Sofia University ``St. Kliment Ohridski'' as part of the Bulgarian National Roadmap for Research Infrastructure.
    } \and 
    Yunfan Wang \thanks{Tsinghua University, Institute for Interdisciplinary Information Sciences, \texttt{yunfan-w20@mails.tsinghua.edu.cn}.}
    }

\begin{document}

\maketitle

\pagenumbering{gobble}

\begin{abstract}
We give the first deterministic algorithm that makes sub-quadratic queries to find the global min-cut of a simple graph in the cut query model. Given an $n$-vertex graph $G$, our algorithm makes $\Otil(n^{5/3})$ queries to compute the global min-cut in $G$. As a key ingredient, we also show an algorithm for finding $s$-$t$ max-flows of size $\Otil(n)$ in $\Otil(n^{5/3})$ queries. We also show efficient cut-query implementations of versions of expander decomposition and isolating cuts, which may be of independent interest.    

\end{abstract}

\clearpage
\tableofcontents
\clearpage

\pagenumbering{arabic}
\section{Introduction}

Computing a global min-cut, also known as edge connectivity of a graph, is one of the most extensively studied problems in algorithmic graph theory  \cite{gomory1961multi,gabow1991matroid,nagamochi1992linear,hao1994faster, karger1994random,karger1996new,karger2000minimum,kawarabayashi2018deterministic,henzinger2020local,li2020deterministic,li2021deterministic,henzinger2024deterministic}. 
In this problem, given an undirected graph $G = (V,E)$, we need to find a smallest set of edges $F\subset E$ such that $G\setminus F$ is not connected.
The problem has also been explored in various computational models, including dynamic algorithms \cite{thorup2007fully,goranci2023fully,jin2024fully}, parallel algorithms \cite{karger1994derandomization,geissmann2018parallel,anderson2023parallel},  distributed algorithms \cite{daga2019distributed,
mukhopadhyay2020weighted,dory2021distributed,ghaffari2022universally}, and streaming algorithms \cite{mukhopadhyay2020weighted,assadi2021simple}.

One model which has attracted significant attention recently is the \emph{cut-query} model. In this model, the edge set $E$ is not known, instead we are allowed to make \emph{cut queries}. A query is a set $S \subseteq V$, and an oracle returns $\Cut(S)$, the number of edges across the cut $(S, V \setminus S)$.

A straightforward algorithm is to learn the entire graph in $\Otil(n^2)$ queries.\footnote{Throughout the paper, we use $\Otil(\cdot)$ to hide a $\polylog(n)$
factor.} 
Rubinstein, Schramm, and Weinberg~\cite{rubinstein2018computing} initiated the study of global min-cut in the cut-query model and showed a cut-query algorithm for simple graphs with only $\Otil(n)$ queries. Mukhopadhyay and Nanongkai~\cite{mukhopadhyay2020weighted} then matched this result for general weighted graphs. Recently, Apers et al.~\cite{apers2022cut} improved this result to only $O(n)$ queries. However, all these algorithms are Monte Carlo randomized algorithms.

Hence, both~\cite{rubinstein2018computing} and~\cite{apers2022cut} posed the question if their algorithms can be made deterministic -- the only known lower bound for deterministic cut queries is $\Omega(n)$~\cite{hajnal1988communication,harvey2008matchings}. 
However,
the best known deterministic algorithm for edge connectivity (or global min-cut) \cite{grebinski2000optimal} still takes $O(n^2/\log n)$ queries to \emph{learn the whole graph}. This only slightly improves upon the trivial algorithm, which makes $\Otil(n^2)$ queries.
Thus, there remains a big gap for deterministic algorithms, and the following conceptual question is still open:

\begin{center}
    \emph{Can we deterministically compute the edge connectivity of a graph\\ 
    without learning the whole graph?}
\end{center}

In this work, we make significant progress towards this problem by showing two results. We show the first deterministic cut-query algorithm for global min-cut in simple graphs with a sub-quadratic number of queries.%

\begin{theorem}~\label{thm:mincut}
Given a simple undirected graph $G$, there is a deterministic algorithm that computes the global minimum cut in $G$ using $\Otil(n^{5/3})$ cut queries.
\end{theorem}

As a key ingredient, we show how to compute $s$-$t$ max-flow using $\Otil(n^{5/3})$ queries.

\begin{theorem}\label{thm:flow}
Given a capacitated undirected graph $G$ with integral edge weights bounded by $W$, and two vertices $s,t \in V$, there exists a deterministic algorithm that computes an $s,t$ max-flow in $\Otil(n^{5/3} W)$ cut-queries. 

\end{theorem}

In particular, \Cref{thm:flow} implies that we can deterministically compute $s$-$t$ max-flow in simple unweighted graphs in $\Otil(n^{5/3})$ cut queries. 
The previous algorithm for the same problem \cite{rubinstein2018computing} obtained this bound using a randomized Monte Carlo algorithm.

Towards obtaining these results we show efficient deterministic cut-query implementation of several tools: a version of expander decomposition, isolating cuts~\cite{li2020deterministic} and a deterministic algorithm that obtains a dominating set of size $\Otil(\frac{n}{\delta})$ in a graph with minimum degree $\delta$. We believe that the cut-query version of these simple and powerful tools, some of which have formed the basis for several fast algorithms in the last decade, may be of independent interest.

    \paragraph{Other Query Models.}
    Query models are widely studied in other settings as well. Many query models can be viewed as learning the properties of a matrix $M$ through vector-matrix-vector $u^T M v$ queries. Different query models and the problems of concern within each model are introduced in \cite{rashtchian2020vector}.
    
    \cite{assadi2021graph} considered a different setting; they studied the query complexity when we cannot use adaptive algorithms and can only perform a limited number of adaptive operations.
    \cite{apers2021sublinear} gave a quantum minimum $s$-$t$ cut algorithms using $\Otil(\sqrt{m} n^{\frac{2}{3}} W^{\frac{1}{3}})$ queries to the adjacency list of $G$, with integral edge weights bounded by $W$.
    \cite{auza2021query} studied the query complexity of the connectivity problem for a graph G under the quantum query model.

\begin{table}[h]
\centering
\renewcommand{\arraystretch}{2}
\begin{tabular}{|c|c|c|c|c|}
\hline
\multirow{2}{*}{} & \multicolumn{2}{c|}{Connectivity} & \multicolumn{2}{c|}{Edge Connectivity} \\ \cline{2-5} 
                      & Lower           & Upper           & Lower              & Upper             \\ \hline

Deterministic         &       \makecell{$\Omega(n) $\\ \cite{hajnal1988communication}   }       &    \makecell{$O(\frac{n\log n}{\log\log n})$\\ \cite{liao2024learning}  }              &   \makecell{$\Omega(n)$ \\ \cite{hajnal1988communication}}          &  \makecell{$\Tilde{O}(n^{5/3})$ \\ \textbf{(This paper)}}      \\ \hline

Zero-error, Randomized            &        \makecell{$\Omega(\frac{n \log \log (n)}{\log n})$\\ \cite{raz1995log} }        &   \makecell{ $O(n)$ \\\cite{apers2022cut}  }           &     $\Omega(n)$           &        \makecell{$\Tilde{O}(n^{5/3})$ \\ \textbf{(This paper)}}              \\ \hline

Bounded Error, Randomized         &     \makecell{$\Omega(\frac{n}{\log n})$\\ \cite{babai1986complexity}}       &    \makecell{$O(n)$ \\ \cite{apers2022cut} }               &       \makecell{$\Omega(\frac{n \log \log (n)}{\log n})$\\ \cite{assadi2021simple}   }        & \makecell{$O(n)$\\ \cite{apers2022cut}}                \\ \hline

Quantum               &        $\Omega(1)$          &    \makecell{$O(\log^5(n))$\\  \cite{auza2021query}}      &           $\Omega(1)$         &     \makecell{$\tilde{O}(\sqrt{n})$ \\ \cite{apers2022cut}}             \\ \hline
\end{tabular}
\caption{The cut query complexity of connectivity and edge connectivity on simple graphs in various models. Before our work, the best deterministic and zero-error upper bound was from \cite{grebinski2000optimal}, they use $O(n^2 / \log n)$ queries by learning the entire graph.  Additionally, the lower bound for zero-error edge connectivity comes from an unpublished article by Troy Lee and Adi Shraibman: ``On the communication complexity of edge connectivity". For more details, see \cite{apers2022cut}.}
\end{table}

\section{Technical Overview}

\paragraph{Maximum $s$-$t$ flow:} In this overview, we assume $W=1$. That is, the input graph $G$ has unit capacities, and the value of $s$-$t$ maximum flow is at most $\Otil(n)$. Our algorithm is based on the classical Dinitz's algorithm~\cite{dinitz2006dinitz}, which we recall below.

The algorithm maintains a flow $f$ and proceeds in iterations. In each iteration, we first construct the (directed) residual graph $G_f$ for the flow $f$. Next, we obtain the layered graph $G_L$ from $G_f$. The vertex set of $G_L$ is $V$. Define layers $L_0, L_1 \ldots $, where $L_i$ is the set of vertices at distance $i$ from $s$ in $G_f$. The edge set of the layered graph is the set of edges $\{u,v\}$ in $E(G_f)$, such that $u \in L_i$ and $v \in L_{i + 1}$ for some $i$. A \emph{blocking flow} $f'$ is then a flow from $s$ to $t$ in $G_L$, such that on every flow path, some edge is saturated (that is, the flow through the edge equals its capacity). Let $\dist(s,t)$ denote the length of the shortest (directed) path from $s$ to $t$ in $G_L$. Finally, we find a blocking flow $f'$ and augment the flow $f$ with $f'$, and proceed to the next iteration. 
The crucial property of the algorithm is that, after each iteration, $\dist(s,t)$ must strictly increase. Once we reach an iteration where a blocking flow of non-zero value does not exist, the obtained flow $f$ is maximum. 
We will use the standard fact that once $\dist(s,t) \geq n^{2/3}$, the total additional units of flow that can be sent in subsequent iterations is at most $O(n^{2/3})$.

To make only sub-quadratic cut queries,  our key contribution is to show how to compute a blocking flow $f'$ using only $\Otil(n + \val(f')\dist(s,t))$ queries where $\val(f')$ denotes the value of $f'$. Our algorithm then keeps augmenting the flow via our blocking flow subroutine until the resulting flow is maximum.

To analyze the query complexity, we divide the algorithm into two phases. The first phase contains all iterations when $\dist(s,t)\le n^{2/3}$, and the second phase contains all subsequent iterations.
In the first phase, we make $\Otil(n\cdot n^{2/3}+n\cdot n^{2/3}) = \Otil(n^{5/3})$ queries because (1) there are at most $n^{2/3}$ iterations since, after that, we must have $\dist(s,t) > n^{2/3}$, (2) the total value of flow augmented during these iterations at most the maximum flow value, which we assumed to be at most $\Otil(n)$, and (3) $\dist(s,t)\le n^{2/3}$.
In the second phase, we again make $\Otil(n\cdot n^{2/3}+ n^{2/3}\cdot n) = \Otil(n^{5/3})$ queries because (1) there are at most $n^{2/3}$ iterations since, after that, $f$ must already be maximum, (2) the total value of flow augmented during these iterations is at most $O(n^{2/3})$, and (3) $\dist(s,t)\le n$.

\paragraph{Global minimum cut:} 
The starting point of our global min-cut algorithm is the algorithm of ~\cite{li2020deterministic}. We give a high-level description of this algorithm below. 

Given a set of terminals $T$, a Steiner min-cut is a minimum cut that separates some pair of terminals, i.e., a minimum size cut $(C, V \setminus C)$ that satisfies $C \cap T$, $T  \setminus C \neq \emptyset$. Note that a Steiner min-cut when $T = V(G)$ is a global min-cut.
Suppose $(C, V \setminus C)$ is the minimum Steiner Cut in the graph $G$. The algorithm of \cite{li2020deterministic} is divided into two cases: one where $\min\{|T \cap C|,|T \setminus C|\} \leq \polylog(n)$, and the other where $\min\{|T \cap C|,|T \setminus C|\} \geq \polylog(n)$. 

In the first case, they find the Steiner min-cut using the \emph{isolating cuts} technique and finish.
In the second case, they compute an expander decomposition and replace $T$ by a smaller set $T' \subseteq T$ with $|T'| \leq \frac{|T|}{2}$, such that some minimum Steiner cut $C$ for $T$ still separates $T'$. Then they recursively run the same procedure with $T'$. Since they do not know which case actually happens,  they run the isolating cut subroutine at every recursion, while continuing to sparsify $T$. 
There can be only $O(\log |T|)$ recursions since the size of $T$ decreases by half after each recursion. 
Observe that if the initial terminal set $T = V(G)$, the algorithm must find a global mincut at some point.

Several challenges arise when implementing this algorithm in the cut-query model. First of all, one needs to compute isolating cuts and expander decomposition in this model. But, superficially, we can be hopeful because the common primitive for these techniques is a $s$-$t$ max-flow algorithm, and we have a new efficient algorithm when the max-flow value is at most $\Otil(n)$.

Unfortunately, the flow instances when computing isolating cuts and expander decomposition may have max-flow value $\Omega(m)$, where $m$ is the number of edges, even when the input graph is simple. For example, the flow instances for computing isolating cuts have the following structure: The super source $s$ connects to many terminals $v$, and the capacity of edge $(s,v)$ is the degree of the vertex $v$. This is similar for the super sink $t$. When $T = V(G)$, for example, the max-flow value of this flow instance can be as large as $\Omega(m)$. 

To overcome the above challenge, our solution is that we choose the initial terminal set $T$ to be, instead of the whole vertex set $V(G)$, a dominating set $D$ of size $\Otil(n/\delta)$, where $\delta$ is the minimum degree.
The correctness of this is based on our new and simple structural lemma (\Cref{lemma:dominating set is t-separated}): every cut $(C, V \setminus C)$ of size at most $\delta-1$ must separate $D$, i.e., we must have $D \cap C, D \setminus C \neq \emptyset$. 
Therefore, assuming that global min-cut has size at most $\delta-1$ (otherwise, it has size $\delta$, which is trivial to compute), we conclude that the minimum Steiner Cut with respect to $D$ is the global min-cut. To initialize our terminal set as $D$, we show how to compute $D$ using $\Otil(n)$ queries.

Why does starting with $D$ help? This is because we will set up all flow instances such that only vertices in $D$ are connected to super source $s$ or super sink $t$ and the capacity of each edge incident to $s$ and $t$ is at most $\delta + 1$.\footnote{\label{footnote:capacity}
Strictly speaking, our max-flow algorithm takes $\Otil(n^{5/3})$ queries only if edges have unit capacity. 
But, we can simulate an edge from super source $s$ to vertex $a$ of capacity $\tau$ by replacing it with $\tau$ paths: $s \to v_1 \to a,\dots,s \to v_{\tau} \to a$. We do similarly for the super sink $t$. The total number of newly added vertices is $\Otil(n/\delta) \cdot (\delta+1) = \Otil(n)$.}

So the max-flow value is only $\Otil(n/\delta)\cdot(\delta+1)=\Otil(n)$, but we prove that this new flow instance can still be used for computing isolating cuts and expander decomposition (in \Cref{sec:isocut,appendix: expander decomposition} respectively).
Therefore, we can employ our max-flow algorithm with $\Otil(n^{5/3})$ queries to compute both isolating cuts and a version of expander decomposition in only $\Otil(n^{5/3})$ queries. 
   
\section{Preliminaries}
In this section, we will formally define terms and notations which we shall adopt throughout the paper.

\paragraph{Graphs:} Given an undirected simple graph $G = (V(G),E(G))$, we define $\delta_G$ to be the \textit{minimum degree} of $G$ and $\lambda_G$ to be the \textit{size of the minimum cut} in $G$. For a vertex set $S\subseteq V$, we let $G[S]$ denote the subgraph induced on the set of vertices $S$. For a set of vertices $S$, we denote by $\partial_G(S)$ the set of edges with exactly one endpoint in $S$. For a vertex $u$, we denote by $N_G(u)$ the set of vertices $v$ with $\{u,v\} \in E(G)$. For a set of vertices $U$, we denote by $N_G(U)$ the set $\bigcup_{u \in U} N_G(u) \setminus U$, the set of neighbors of vertices in $U$ outside $U$. We denote a cut in $G$ using the notation $(S, V(G) \setminus S)$ for $S \subseteq V(G)$. For simplicity, we shall refer to $V(G) \setminus S$ as $\overline{S}$.
We say that a vertex set $S\subseteq V(G)$ is a \textit{dominating set} of $G$ if, for any vertex $v \in V(G)$, either $v \in S$ or $v$ has at least one neighbor in $S$.

Given a directed graph $H$, we use $\dist_H(s,t)$
to represent the distance  (length of the shortest path) between two vertices $s,t \in V(H)$.  We drop the subscripts when the graph is clear from the context. 

\paragraph{Flows and residual graphs:}For a flow $f$ from $s$ to $t$, we use $\val(f)$ to denote the \textit{size of the flow} $f$. We describe the flow by flow values $f(u,v),f(v,u)$ for every edge $\{u,v\} \in E$ with $f(u,v) = -f(v,u)$.
Given a graph $G$ with non-negative edge capacities $c: E \rightarrow \mathbb{R}_{\ge 0}$ and a flow $f$, we define the residual graph as the directed capacitated graph $G_f$, whose vertex set is $V(G_f) = V$. For an edge $\{u,v\} \in E$ with flow $0 \leq f(u,v) \leq c(u,v)$, the \textit{residual capacity} is defined as $c_f(u,v) = c(u,v) - f(u,v) \ge 0$. We also define $c_f(v,u) = c(v,u) + f(u,v)$.
Then we add two edges in $G_f$: if $c_f(u,v) > 0$, we add an edge  from $u$ to $v$ with capacity $c_f(u,v)$ and if $c_f(v,u) > 0$, we add an edge from $v$ to $u$ with capacity $c_f(v,u)$.

\paragraph{Types of queries:}
Next, we formally define cut queries and bipartite independent set queries.
\begin{definition} [Cut Query]
For an undirected capacitated graph $G$ with capacities $c:E(G) \rightarrow \mathbb{R^+} \cup \{0\}$, given a set $S \subseteq V(G)$, the cut query oracle for graph $G$ returns the total capacity of cut-edges $\Cut(S) :=c(S,\overline{S}) = \sum_{a \in S, b \in \overline{S}} c(a,b)$.
\end{definition}

\begin{definition}[Bipartite Independent Set (BIS) Query]
\label{definition: BIS}
    Given a (possibly directed) capacitated graph $H$ with capacities $c: E(H) \rightarrow \mathbb{R^+}\cup \{0\}$ , sets $A, B \subseteq V(H)$ as input such that $A \cap B =\emptyset$, the bipartite independent set query oracle returns the boolean value which indicates if $c(A,B) = \sum_{a \in A, b \in B} c(a,b) > 0$. 
\end{definition}

BIS queries were introduced by~\cite{beame2018edge} and have been extensively studied in the context of problems on undirected graphs.
In an undirected graph $H$, a BIS query between $A$ and $B$ can be simulated with $3$ cut queries since 
\begin{align*}
    c(A,B) 
    &= \frac{1}{2} \left( \Cut(A) + \Cut(B) - \Cut(A \cup B) \right)
\end{align*}
We will require BIS queries on the residual graph for our maximum flow algorithm. The following lemma demonstrates how to simulate BIS queries on the residual graph using cut queries.

\begin{lemma}
\label{lemma:flowBIS}
    Given an undirected capacitated graph $G$ and an explicit flow $f$, we can simulate a BIS query on the residual graph $G_f$ by using $3$ cut queries in $G$.
\end{lemma}
\begin{proof}
    Given $A, B \subseteq V$, the BIS query on $G_f$ will return whether there exists $a \in A$, $b \in B$ such that $c_f(a,b) > 0$. This is equivalent to
    \[ \sum_{a\in A} \sum_{b \in B} c_f(a, b) > 0\] 
   Notice that
    \begin{align*}
        \sum_{a\in A} \sum_{b \in B} c_f(a, b) &= \sum_{a\in A} \sum_{b \in B} c(a, b) - f(a,b) \\
        &= \sum_{a\in A} \sum_{b \in B} c(a, b) - \sum_{a\in A} \sum_{b \in B} f(a, b)\\
        &= \frac{1}{2}[\Cut(A) + \Cut(B) - \Cut(A\cup B)] - \sum_{a\in A} \sum_{b \in B} f(a, b)\\
    \end{align*}
    We can use $3$ cut queries to get $\Cut(A)$, $\Cut(B)$ and $\Cut(A\cup B)$. Since $f$ is explicit, we can obtain $\sum_{a\in A} \sum_{b \in B} f(a, b)$ in zero queries. Therefore, a BIS query in $G_f$ can be simulated by $3$ cut queries in $G$.
\end{proof}

\section{Flow Algorithm}

In this section, we prove our result on computing an $s$-$t$ max flow of bounded size (\Cref{thm:flow}).
~\Cref{subsection:BIS} describes how to obtain simple primitives of the graph using BIS queries.~\Cref{subsection:Dinitz} recalls Dinitz's algorithm and its analysis.~\Cref{subsection:implement} describes how to implement Dinitz's algorithm via BIS queries.

\subsection{Simple Primitives using BIS Queries}\label{subsection:BIS}
\begin{lemma}
\label{lemma: find a neighbor}
Given a graph $G$ and a  flow $f$ in $G$, for two disjoint subsets $A,B\subseteq V(G)$, if $E_{G_f}(A,B)\ne \emptyset$, we can find a neighbor of $A$ in $B$ (a vertex $b \in B$ such that there exists a vertex $a \in A$ with $(a,b) \in E(G_f)$) in the residual graph $G_f$ in  $O(\log |B|)$ BIS queries on $G_f$.
\end{lemma}
\begin{proof}
    Do a binary search on $B$. Divide $B$ into two parts $B_1, B_2$ with roughly equal size, i.e. $|B_1| \approx |B_2|$. 
    We know that either $E_{G_f}(A, B_1) \ne \emptyset$ or $E_{G_f}(A, B_2) \ne \emptyset$ because $E_{G_f}(A,B)\neq \emptyset$.
    We check which one of these is true using one BIS query. 
    If $E_{G_f}(A, B_1) \ne \emptyset$, then we recursively find a neighbor in $B_1$. Otherwise, we recurse on $B_2$. So, we use  $\mathcal{O}(\log |B|)$ BIS queries on $G_f$.
\end{proof}

\begin{corollary}
\label{corollary: neighbor finding}
    For a vertex $u \in V$ or a set of vertices $U\subseteq V$, we can learn $N_{G_f}(u)$ and $N_{G_f}(U)$ in $\Otil(|N_{G_f}(u)|)$ and $\Otil(|N_{G_f}(U)|)$ BIS queries on $G_f$, respectively.
\end{corollary}
\begin{proof}
We first show this for a single vertex $u$. Let $B = V \setminus \{u\}$. By  \Cref{lemma: find a neighbor} 
we can find a neighbor of $u$ in $B$ in $G_f$ using $O(\log n)$ BIS queries. Next, delete $v$ by setting $B = B \setminus \{v\}$, and invoke~\Cref{lemma: find a neighbor} to find another neighbor of $u$. Inductively, we learn the set $N_{G_f}(u)$ using $O(|N_{G_f}(u)|\log n)$ BIS queries. 

The argument for a set of vertices $U$ is obtained by treating the set $U$ as one super vertex.
\end{proof}

\begin{corollary}
\label{corollary: BFS tree finding}
    Given graph $G$ and a flow $f$, we can find a BFS tree in $G_f$ with root $s$ in $\Otil(n)$ BIS queries on $G_f$.
\end{corollary}
\begin{proof}
     We build the tree $T$ in layers. Let $V(T) = \{s\}$, and $E(T) = \emptyset$.     
     Additionally, we will maintain the distance from $s$ for each vertex $u \in V$, denoted by $d(u)$. We initialize $d(s) = 0$, and $d(u) = \infty$ for every other vertex. 
     
     In every iteration, we find a vertex in $V(T)$ with the smallest distance to $s$, say $u$. We find all neighbors $u'$ of $u$ in $V \setminus V(T)$ using \Cref{corollary: neighbor finding}. For each such neighbor $u'$ of $u$, we update the distance as $d(u') = d(u) + 1$, and add $u'$ to $V(T)$ and the edge $(u,u')$ to $E(T)$. Let $x$ be the number of neighbors of $u$ (in $G_f$) in $V \setminus V(T)$. Then we spend $\Otil(x)$ queries for this step. We charge these queries to the $x$ vertices which are added to $V(T)$ in this step. Repeating the same process for at most $n - 1$ iterations, we correctly obtain the BFS tree $T$.

     Note that each vertex $v \in V(T)$ receives $O(1)$ charge, and the number of vertices in $T$ is $n$. Therefore, the total query complexity is $\Otil(n)$.
\end{proof}

\subsection{Reminder of  Dinitz's Algorithm}\label{subsection:Dinitz}

Here we recall Dinitz's blocking flow algorithm \cite{dinitz2006dinitz} and its analysis.

\begin{definition} [Layered Graph]
    Let $G$ be an undirected graph and let $s$ and $t$ be source and sink vertices in $V(G)$. Further, suppose that $f$ is a $s$-$t$ flow in $G$, and let $G_f$ be the residual graph of $G$ with respect to $f$. The layered graph $G_L$ of $G$ with respect to $f$ is the decomposition of the vertices of $G$ into layers defined by their distance from the source in $G_f$. The $i$-th layer is defined as $L_i = \{v \mid \dist_{G_f}(s,v) = i\}$. The vertex set of the layered graph is $V(G_L) = V(G)$, and the edge set is the set of edges $(u,v)$ of $G_f$ such that $u \in L_i$ and $v \in L_{i + 1}$ for some layer $i$. 
\end{definition}

Note that we always have $L_0 = \{s\}$. For simplicity, we will assume that the last layer contains only the sink vertex $t$ - if there are other vertices in this layer, we will drop them.

\begin{definition} [Blocking Flow]
Given a graph $G$, the source and sink vertices vertex $s,t$, and an $s$-$t$ flow $f$ in $G$, let $G_L$ be the layered graph $G_L$ of $G$ with respect to $f$. An $s$-$t$ flow $f'$ in $G_L$ is called a blocking flow, if for every $s$-$t$ path in $G_L$, some edge is saturated by $f'$

\end{definition}

\begin{algorithm}[H]
\caption{Dinitz's blocking flow algorithm\cite{dinitz2006dinitz}}
\label{algorithm: Dinitz's blocking flow}
\begin{algorithmic}[1]
\Require Graph $G$, two vertices $s$ and $t$
\Ensure An $s$-$t$ maximum flow $f$ in $G$
\State Initialize an empty flow $f$ from $s$ to $t$.
\State Construct the residual graph $G_f$ from $G$ and $f$
\State Construct the layered graph $G_L$ from $G_f$ and the source vertex $s$.
\If{$\text{dist}_{G_L}(s,t) > n$}
\State Terminate and output the flow $f$.
\Else
\State Find a blocking flow $f'$ from $s$ to $t$ in $G_L$
\State Augment $f$ by $f'$ to obtain a new flow $f''$. Update $f \leftarrow f''$, and repeat step $2$.
\EndIf
\end{algorithmic}
\end{algorithm}

The next lemma gives some simple facts about maximum flows which will be useful for the subsequent analysis. 
\begin{fact}\label{lemma:fact}   We have the following:
\begin{enumerate}
\item A flow $f$ is maximum if and only if there is no augmenting path in the residual graph $G_f$. 
\item If there is an $s$-$t$ cut of capacity $B$ in $G_f$, then 

the maximum flow from $s$ to $t$ in $G$ is at most $\val(f) + B$.
\end{enumerate}
\end{fact}

Next, we need the following property of Dinitz's algorithm (\Cref{algorithm: Dinitz's blocking flow}).

\begin{lemma}[\cite{dinitz2006dinitz}]\label{lemma:distance}
After augmenting the flow $f$ with a blocking flow $f'$ to obtain a new flow $f''$, the distance from $s$ to $t$ in the residual graph increases by at least $1$. Concretely, $\dist_{G_{f''}}(s,t) \geq \dist_{G_{f}}(s,t) + 1$.
\end{lemma}

Thus Dinitz's algorithm must terminate after $n$ iterations of blocking flow and give us a maximum flow, as after these many iterations, there cannot be an $s$-$t$ path of length $n$. However, this is too expensive for us. Instead, we will use a standard optimization that actually shows that the algorithm terminates in $O(n^{2/3}W)$ iterations. We provide the proof here for completeness.

\begin{lemma}\label{lemma:dinitz_optimize}
Given a capacitated graph $G$ where each edge weight is an integer in the range $[1,W]$ and two vertices $s,t$,~\Cref{algorithm: Dinitz's blocking flow} finds an $s$-$t$ maximum flow after $O(n^{2/3}W)$ iterations of blocking flow.
\end{lemma}

 \begin{proof}

     For the purpose of analysis, consider the $s$-$t$ flow $f_{i}$ obtained after $i$ iterations, where $i$ is the smallest subscript such that $\dist_{G_{f_i}}(s,t) > d := n^{2/3} $. By~\Cref{lemma:distance}, since the distance between $s$ and $t$ in the residual graph after $i$ iterations exceeds $d$, it follows that $i \le d$. This bounds the number of iterations required.
    
    Since the graph contains $n$ vertices and there are at least $d$ layers in the layered graph $G_L$, by the pigeonhole principle, there must exist a pair of consecutive layers $L_{j}$ and $L_{j+1}$ such that $|L_{j}| + |L_{j+1}| \le O(\frac{n}{d})$. Therefore, the total edge weight between these layers is bounded by $O((\frac{n}{d})^2 W) = O(n^{2/3}W)$. 
    By the definition of the layered graph, the set of edges $E(L_j, L_{j+1})$ forms a cut between $s$ and $t$ in $G_{f_i}$. Thus, by \Cref{lemma:fact}, the total flow from $s$ to $t$ is at most  $\val(f_i) + O(n^{2/3}W)$. 
    
    Consequently, the remaining flow size in the residual graph $G_{f_i}$ is at most $O(n^{2/3}W)$. Since each iteration of blocking flow increases the flow value by at least $1$, it follows that the algorithm can perform at most $O(n^{2/3}W)$ additional iterations before termination. Therefore, the algorithm terminates after at most $O(n^{2/3}W)$ iterations in total.
\end{proof}

\subsection{Implementing Dinitz's Algorithm via BIS Queries}\label{subsection:implement}
Now we implement Dinitz's blocking flow algorithm \cite{dinitz2006dinitz} in the cut query model.

    \begin{claim}
    \label{claim: single round blocking flow}
    Given a flow $f$ such that $\dist_{G_f}(s,t) = d'$,  we can compute a blocking flow $f'$, augment the flow $f$ with $f'$, and recompute the residual graph using $\Otil(n + d' \cdot \val(f'))$ BIS queries on $G_f$.
    \end{claim}

    \begin{proof}
        We first construct the layered graph $G_L$ of $G_f$ by running BFS in $\Otil(n)$ queries using Corollary \ref{corollary: BFS tree finding}. This gives us $L_0=\{s\}, ..., L_{d'} = \{t\}$. The goal is to find an $s$-$t$ path in $G_L$ efficiently using a stack-based method. 
        
        Next, we maintain a stack to track our search for an $s$-$t$ path. Initially, the stack contains $\{s\}$. At each step, let the top element of the stack be $u$, where $u \in L_i$. We then attempt to find $v \in L_{i+1}$ such that $c_f(u,v) > 0$. If such $v$ exists, we can find it using $\Otil(1)$ queries by Lemma \ref{lemma: find a neighbor}, and we push $v$ onto the stack. 
        If no such $v$ exists, we pop $u$ from the stack and remove $u$ from the layered graph, as $u$ cannot send any flow to the next layer. 

        When the stack reaches a length of $d'+1$, this indicates that we have found an $s$-$t$ path. We add this flow path to $f'$, adjust the residual  capacities along the path, and then reset the stack to $\{s\}$ to search for another path.
        The process terminates when the stack is empty, meaning there is no further flow from $s$ to any vertex in the next layer.

        Now we analyze the query complexity. In each step, we either push a new vertex onto the stack or pop one off and remove it from the graph. Each push or pop requires $\Otil(1)$ queries. Whenever the stack reaches length $d' + 1$, we send flow from $s$ to $t$, and the stack is reset, reducing its length by $d'$. 
        The total length of the stack can increase by at most $n + O(d') \cdot \val(f')$, as this is the maximum number of vertices and flow pushes in the layered graph. Therefore, the total number of queries is at most $\Otil(n + d' \cdot \val(f') )$.
    \end{proof}

We are now ready to prove the main result,~\Cref{thm:flow}.

\begin{proof}[Proof of~\Cref{thm:flow}]
    By~\Cref{claim: single round blocking flow}, a single blocking flow iteration requires $\Otil(n + d' \cdot \val(f'))$ queries, where $\val(f')$ represents the current flow value, and $d'$ is the distance from $s$ to $t$ in the layered graph. Therefore, to analyze the algorithm's query complexity, we should consider three factors: the total number of iterations, the flow value, and the distance. Balancing these factors is essential for optimizing the overall query complexity.

    By \Cref{lemma:dinitz_optimize}, we know that the number of iterations is at most $O(n^{2/3}W)$. When the distance is less than $n^2/3$, the flow value is at most $O(nW)$. Conversely, when the distance exceeds $n^{2/3}$, the remaining flow value is at most $O(n^{2/3} W)$.

    Thus, the total query complexity can be expressed as:
    \[\sum \Otil(n + d'\cdot \val(f')) \le O(n^{\frac{2}{3}} W) \cdot \Otil(n)  + \sum \Otil(d'\cdot \val(f'))\]

    Breaking it down further:   
    \[ = \Otil(n^{\frac{5}{3}} W) + \sum_{d' < n^{\frac{2}{3}}} \Otil(d'\cdot \val(f')) + \sum_{d' \ge n^{\frac{2}{3}}} \Otil(d'\cdot \val(f'))\]

    For the two cases:
    \begin{enumerate}
        \item For $d' < n^{\frac{2}{3}}$, we have 
            \[\sum_{d' < n^{\frac{2}{3}}} \Otil(d'\cdot \val(f')) \le n^{\frac{2}{3}} \cdot O(nW) = \Otil(n^{\frac{5}{3}} W)\]
        \item For $d' \ge n^{\frac{2}{3}}$, we have 
            \[\sum_{d' \ge n^{\frac{2}{3}}} \Otil(d'\cdot \val(f')) \le n \cdot O(n^{\frac{2}{3}} W) = \Otil(n^{\frac{5}{3}} W)\]
    \end{enumerate}

    It now follows that the total query complexity is at most $\Otil(n^{5/3} W)$, which proves the result.
    \end{proof}

    \section{Global Min-cut Algorithm}
    In this section, we present our algorithm for finding global min-cut, proving~\Cref{thm:mincut}. We will focus on obtaining the following threshold version of the result.
    \begin{theorem}
\label{theorem: main theorem}
    Given a graph $G$ and a parameter $\tau \le \delta-1$, there is an algorithm that either
    \begin{itemize}
        \item returns a cut $(C, V \setminus C)$ with size $|\partial C|$ at most $\tau$, or
        \item certifies that the global min-cut in $G$ must have a size larger than $\tau$
    \end{itemize}
    in $\Otil(n^{5/3})$ cut queries.
\end{theorem}

Note that \Cref{thm:mincut} now follows immediately by binary search. If the global min-cut has size $\delta$, we simply return the cut corresponding to of the minimum-degree vertex.
We fix the parameter $\tau$ throughout this section.

The rest of this section is dedicated to proving~\Cref{theorem: main theorem}. 
We first give a crucial yet simple observation that any dominating set must be separated by every cut of size at most $\delta - 1$ in \Cref{subsection:dominating set}, and show how to efficiently compute the dominating set in \Cref{sec:domset alg}. 
Then, we show how to efficiently compute minimum isolating cuts of size at most $\tau$ in the cut query model in \Cref{sec:isocut}.

Given these three subsections, we are ready to implement the global minimum cut algorithm of \cite{li2020deterministic} in the cut query model. In high level, the algorithm starts by computing a dominating set $R$ and has two cases. If the minimum cut is unbalanced with respect to $R$, then we will compute minimum isolating cuts with respect to subsets of $R$ (\Cref{subsection: unbalanced case}). Otherwise, if the minimum cut is balanced with respect to $R$, then we will compute expander decomposition with respect to $R$ in the cut query model (\Cref{appendix: expander decomposition}) and use the decomposition to ``sparsify'' $R$ to be a smaller subset $R' \subseteq R$ that is still separated by minimum cuts (\Cref{subsection: balanced case}).

\subsection{Dominating Sets are Separated by Small Cuts}\label{subsection:dominating set}
Below, we show a simple and crucial observation: any dominating set must be separated by a cut of size at most $\delta-1$.

\begin{definition}
\label{definition: t separate}
    Given a graph $G$, a vertex  $R \subseteq V$ is $c$-separated if for every cut $(C, V \setminus C)$ of size  $|\partial C| \leq c$, we have $C \cap R, R \setminus C \neq \emptyset$. That is, $R$ hits both sides of the cut $C$.
\end{definition}

\begin{lemma}\label{lemma:dominating set is t-separated}
    For any dominating set $R\subseteq V$, $R$ is $\delta-1$-separated, where $\delta$ is the minimum degree.
\end{lemma}

\begin{proof}
    Suppose for contradiction that this is not the case. Then there is a cut $(C_1, C_2 = V \setminus C_1)$ of size at most  $\delta-1$ where  $R \subseteq C_1$. Now we have $\delta - 1 \ge |E(C_1, C_2)| \ge |E(R, C_2)| \ge |C_2|$ where the last inequality is because $R$ is dominating. But every vertex in $C_2$ has degree at least $\delta$, therefore we must have $|E(C_1, C_2)| \ge |C_2|(\delta - (|C_2| - 1)) \geq \delta$ for $1 \leq |C_2| \leq \delta - 1$, which leads to a contradiction.
\end{proof}

We remark that the above lemma is the crucial place where we exploit that the graph is simple.

\subsection{Computing a Dominating Set}
\label{sec:domset alg}

In this section, we show how to compute a dominating set efficiently in the cut-query model.

\begin{theorem}
\label{theorem: find a dominating set}
    There is a deterministic algorithm that finds a dominating set $R\subseteq V$ of size  $|R| \le \Otil(\frac{n}{\delta})$
    using $\Otil(n)$ cut-queries.
\end{theorem}

\begin{proof}
We will first state our algorithm, then prove its correctness, and finally analyze its complexity under the cut query model.

\paragraph*{Algorithm}
    We construct $R$ iteratively. Initially, we set $G' = G$, $R = \emptyset$. We first apply the following reduction rule -- as long as there is a vertex in $v \in G'$ whose degree in $G'$ is larger than $\frac{\delta}{2}$, we add it into $R$ and delete $v \cup N_{G'}(v)$ from $G'$.

    After exhaustively applying the above reduction rule, if $V(G') \ne \emptyset$ and there is no vertex in $G'$ with a degree larger than $\frac{\delta}{2}$, then we do the following. Consider the bipartite graph $H$ whose vertex set is the bipartition $W_1 \cup W_2$, where $W_1 = N_G(R)$ and $W_2 = V(G')$. The edge set of $H$ is the edge set of $G$ restricted to the edges between $W_1$ and $W_2$.
    
    First, observe that $\forall u \in W_2$, $|E(W_1, u)| = \deg(u) - |E(u, W_2 \setminus \{u\})| \ge \frac{\delta}{2}$, otherwise, this vertex would have been considered in the previous step. Thus, we have $|E(W_1, W_2)| = \sum_{u \in W_2} |E(W_1, u)| \ge \sum_{u \in W_2} \frac{\delta}{2} \geq |W_2|\frac{\delta}{2}$. Then there exists a $v \in W_1$ such that $|E(v, W_2)| \ge \frac{|E(W_1,W_2))|}{|W_1|} \ge \frac{|W_2|\delta}{2|W_1|} \ge \frac{|W_2| \delta}{2n}$. 
    We add $v$ into $R$ and delete all its neighbors in $G'$. 
    We repeat until $V(G') = \emptyset$, hence $R$ must be a dominating set.

    \paragraph*{Correctness}
    It remains to  bound the size of $R$. When there is a vertex in $G'$ with degree larger than $\frac{\delta}{2}$, we reduce $|V(G')|$ by at least $\frac{\delta}{2}$ each time. So we can add at most $\frac{2n}{\delta}$ such vertices into $R$.
    
    When no vertex in $G$ has degree larger than $\frac{\delta}{2}$, we remove at least $|V(G')|\frac{\delta}{2n}$ vertices from $G'$. Thus the size of $G'$ reduces by at least a $1 - \frac{\delta}{2n}$ fraction. Thus after ${O}(\frac{n}{\delta})$ iterations, the size of $G'$ reduces by a constant factor, and hence this step can happen at most ${O}(\frac{n}{\delta} \log n) = \Otil(\frac{n}{\delta})$ many times. In total, we have $|R| \le \frac{2n}{\delta} + O(\frac{n}{\delta} \log (\frac{n}{\delta})) = \Otil(\frac{n}{\delta})$.
    
    \paragraph*{Complexity}
    First, we want to find a vertex $v$ whose degree in $G'$ is at least $\frac{\delta}{2}$. Given a vertex $w$, we can find the degree of $w$ in $G'$ in $\mathcal{O}(1)$ queries. We enumerate each vertex $w \in V(G')$, and find its degree $\deg(w)$. If $\deg(w) < \frac{\delta}{2}$, then we mark $w$ as irrelevant and continue. If $\deg(w) > \frac{\delta}{2}$, then we find all its neighbors in $G'$ using $\Otil(\deg(w))$ queries using~\Cref{lemma: find a neighbor}. Note that we delete all these neighbors of $w$ from $G'$. Let us charge the query cost uniformly to each vertex in $\deg(w)$, and note that each vertex receives a charge of $\Otil(1)$. Observe that the degree of a vertex in $G'$ is non-increasing, thus irrelevant vertices remain irrelevant.
    
    Now it remains to analyze the case when no vertex of $G'$ has degree $> \frac{\delta}{2}$ in $G'$. Recall that we want to find a vertex $v \in W_1$ such that $|E(v, W_2)| \ge \frac{|E(W_1,W_2))|}{|W_1|}$.
    We will use binary search. Divide $W_1$ into two (roughly) equal sized groups $W_{11}$ and $W_{12}$. By~\Cref{lemma: find a neighbor}, using $\Otil(1)$ queries, we can find both $E(W_{11}, W_2)$ and $E(W_{12}, W_2)$. By simple averaging, either ${E(W_{11}, W_2)} \geq \frac{|E(W_1,W_2))|}{|W_{1}|}{|W_{11}|}$ or ${E(W_{12}, W_2)} \geq \frac{|E(W_1,W_2))|}{|W_{1}|}{|W_{12}|}$. Without loss of generality, assume the former. We now recurse on $W_{11}$, and repeat the same process, till we reach a singleton vertex $v \in W_1$, which must satisfy $|E(v, W_2)| \ge \frac{|E(W_1,W_2))|}{|W_1|}$. 
    
    The total number of cut queries is at most $\Otil(1)$, since in each iteration we use $\Otil(1)$ queries and there are at most $\Otil(1)$ iterations of the binary search.
    Therefore, the total query complexity of finding the dominating set is dominated by the query complexity of the reduction rule, and hence is at most $\Otil(n)$.
\end{proof}

\subsection{Computing Minimum Isolating Cuts}
\label{sec:isocut}
This section shows an efficient algorithm for computing a minimum isolating cut. Let us first recall its definition below.

\begin{definition} [Minimum Isolating Cut]
    For any set of vertices $R \subseteq V$ and $r \in R$, the minimum isolating cut of $r$ is an $\{r\}$-$(R\setminus\{r\})$ min-cut. The minimum isolating cut (of $R$) is the minimum sized cut among the minimum isolating cuts over all $r\in R$.
\end{definition}

\begin{theorem}
\label{theorem: isolating cut}
    Given an unweighted graph $G$ and a set of vertices $R\subseteq V$ such that $|R| $ is at most $ \Otil(\frac{n}{\tau})$, there is an algorithm that either
    \begin{itemize}
    \item outputs a minimum isolating cut of $R$ of size at most $\tau$, or
        \item certifies that the minimum isolating cut of $R$ has a size larger than $\tau$
    \end{itemize}
    in $\Otil(n^{5/3})$ cut queries.
\end{theorem}

We remark that the algorithm can be modified easily to output an isolating cut for each $r \in R$ whose min-isolating cut has size at most $\tau$.
\begin{proof}

First, we will state the algorithm, which follows the approach of~\cite{li2020deterministic}, adapted to our setting. Then we will discuss its correctness, and finally argue about implementation and the query cost. We will use the concept of \emph{closest min-cuts}. For two sets $S$ and $T$, we say that an $S$-$T$ min-cut $(X,V \setminus X)$ is closest to $S$ if for any other $S$-$T$ min-cut $(X', V \setminus X')$, we have $X \subseteq X'$.  

    \paragraph{Algorithm:} 
    Encode each vertex in $R$ using a unique $\mathcal{O}(\log |R|)$ dimensional bit vector. Next, compute $O(\log |R|)$ different bi-partitions of the vertices in $R$, where each bi-partition is obtained by selecting a coordinate $c$ of the bit vector and grouping the vertices based on whether they have $0$ or $1$ in the $c^{\text{th}}$ coordinate. Let $\mathcal{C}$ be the set of these bi-partitions. Note that each pair of distinct vertices $u,v \in R$ will be separated by at least one bi-partition in $\mathcal{C}$.

    For each of these $O(\log |R|)$ bi-partitions $(A, B) \in \CC$ of $R$, $A \cap B = \emptyset$, set up the following flow problem. Add two vertices $s_{\text{source}}, s_{\text{sink}}$ to the vertex set of $G$. For each vertex $a\in A$, add an edge $(s_{\text{source}}, a)$ with capacity $\tau+1$. Similarly, for each $b\in B$, add an edge $(b, s_{\text{sink}})$ with capacity $\tau+1$. Every edge of $G$ has a capacity of $1$. Call the modified graph $H$. Find an $s_{\text{source}}$-$s_{\text{sink}}$ max-flow in $H$.

    Let $f_{A, B}$ be this max-flow and $(C_A, C_B = V(G) \setminus C_A)$ be the restriction of the minimum cut corresponding to this maximum flow to $G$ (we simply remove $s$ and $t$ from the cut to obtain $C_A, C_B$). Consider the graph $G'$ after deleting every such cut, that is, we delete the edges $\partial{C_A}$ for every bi-partition $(A, B) \in \CC$. For each $r \in R$, define $T_{r} \subseteq V(G)$ to be the set of vertices reachable from $r$ in $G'$. Further, let $R' \subseteq R$  be the set of $r \in R$ which satisfy the following property -- for every bi-partition $(A, B) \in \CC$, if $r \in A$, then $r \in C_A$, and if $r \in B$, then $r \in C_B$. In other words, $R'$ is the subset of $R$ whose edges to $s_{\text{source}}$ or $s_{\text{sink}}$ are not saturated in any max-flow $f_{A,B}$ across all $(A,B) \in \mathcal{C}$ (see~\Cref{claim:in T(A)}).

    Let $T_{r}' = \{r\} \cup (T_{r} \setminus R)$ 
    for each $r \in R'$, we obtain a new capacitated graph $G_{r}$ from $G$ as follows. %
    Contract all vertices of $V\setminus T_{r}'$ into vertex $s_r$ while keeping parallel edges.
    Then, we compute the min-cut between $r$ and $s_r$ in $G_r$. We check if it is an isolating cut for $r$. Let $\lambda_r$ be the size of this cut (and let $\lambda_r = \infty$ if no such cut exists). Then we check if there exists an $r \in R'$ with $\lambda_r \leq \tau$. If so, we output the corresponding cut. Otherwise, we declare that the minimum isolating cut for $R$ has a size larger than $\tau$. This concludes the description of the algorithm.

    \paragraph{Correctness:} We now show the correctness of our algorithm. Note that our algorithm is similar to~\cite{li2020deterministic}, but we need to argue its correctness a bit more carefully since whenever we compute an $A$,$B$ min-cut for a bi-partition, we set the capacities of the edges incident to the source and sink vertices as $\tau + 1$ (instead of $\infty$ as in their setting). On a high level, this still works since we are only interested in the isolating cuts of size at most $\tau$; if there is no isolating cut of size $\leq \tau$, we will simply output that this is the case. The next few lemmas formally show that our algorithm is indeed correct.

    \begin{claim}\label{claim:in T(A)}
        For any bipartition $(A,B) \in \mathcal{C}$ and any $a \in A$, we say that $a$ is saturated if in the maximum flow $f_{A,B}$, the edge $(s_{\text{source}}, a)$ is saturated.  Then if $a$ is not saturated by $f_{A, B}$, then $a \in C_A$. The same holds analogously for any $b \in B$. 
    \end{claim}
    \begin{proof}
        If $a \notin C_A$, then the edge $(s_{\text{source}}, a)$ is cut by the $(s_{\text{source}}, s_{\text{sink}})$ min-cut. But since a min-cut is saturated in any max-flow, this means that $a$ must be saturated, which is a contradiction. 
    \end{proof}

    \begin{claim}\label{claim:mincut of unsat}
            For each bi-partition $(A,B)$, $(C_A, C_B)$ is a minimum cut in $G$ that separates $C_A \cap A$ and $C_B \cap B$. 
        
    \end{claim}

    \begin{proof}
    Since $(C_A,C_B)$ is a restriction of the $s_{\text{source}}$-$s_{\text{sink}}$ min-cut, it follows that there exists a feasible flow in $G$ from vertices of $C_A \cap A$ to vertices of $C_B \cap B$ saturating every edge of $\partial{C_A}$. This flow certifies that $(C_A, C_B)$ must be a minimum cut between $C_A \cap A$ and $C_B \cap B$ in $G$.
    \end{proof}

    \begin{claim}
    \label{claim: no crossing}
    Consider some terminal $r \in R$ and let $(C_r, V \setminus C_r)$ be the min-cut between $\{r\}$ and $R \setminus \{r\}$ where $r \in C_r$, which is closest to $r$ (that is, the closest isolating cut for $r$). If $|E(C_r, V \setminus C_r)| \leq \tau$,  then the following statements must hold for every bipartition $(A,B) \in \mathcal{C}$.

    \begin{itemize}
    \item $r$ cannot be saturated in the flow $f_{A,B}$. 
    \item If $r \in A$, $C_r \subseteq C_A$. Likewise, if $r \in B$, we must have $C_r \subseteq C_B$.

    \end{itemize}
    \end{claim}
    \begin{proof}
    Consider a bi-partition $(A,B) \in \mathcal{C}$. Let us assume without loss of generality that $r \in A$, the other case is symmetric. Since $(C_r, V \setminus C_r)$ is a cut of size at most $\tau$ between $r$ and $R \setminus \{r\}$, $r$ cannot be saturated. Therefore, by \Cref{claim:in T(A)}, we have $r \in C_A$. 

    Now, we will show that $C_r \subseteq C_A$. Suppose, for the sake of contradiction, that this is not the case. Then by submodularity of cuts, we have:
    \[ \partial(C_r) + \partial(C_A) \ge \partial(C_r \cap C_A) + \partial(C_r \cup C_A). \]
    Since $C_r$ is the minimum cut separating $r$ and $R\setminus \{r\}$, it follows that $\partial(C_r) \le \partial(C_r \cap C_A)$. Similarly, since $(C_A, C_B = V \setminus C_A)$ is the minimum cut separating $C_A \cap A$ and $C_B \cap B$ (see \Cref{claim:mincut of unsat}), we must have $\partial(C_A) \le \partial(C_r \cup C_A)$.
    Combining these results, we obtain the equalities:
    \begin{center}
        $\partial(C_r) = \partial(C_r \cap C_A)$ and $\partial(C_A) = \partial(C_r \cup C_A)$.
    \end{center}
     This implies that $C_r \cap C(A)$, which is a subset of $C_r$, is also a minimum cut separating $r$ from $R \setminus \{r\}$. This contradicts the fact that $C_r$ is the minimum cut between $r$ and $R \setminus \{r\}$ that is closest to $r$.
    \end{proof}
    
    The following lemma shows that across all vertices $r \in R'$, the sets $T_{r}$ are disjoint.
    
    \begin{lemma}
    \label{lemma: separation}
        For every distinct vertices $r, r' \in R'$, we have $T_{r} \cap T_{r'} = \emptyset$,  $r\in T_{r}$, and $r' \in T_{r'}$. 
    \end{lemma}
    \begin{proof}
        Recall that each pair of distinct vertices in $R'$ is separated by at least one of the bi-partitions defined by bit vectors. Without loss of generality, we can assume that there exists a bi-partition $(A, B) \in \mathcal{C}$ such that $r \in A$ and $r' \in B$. Since $r$ and $r'$ are not saturated, by~\Cref{claim:in T(A)}, we have $r \in C_A$ and $r' \in C_B$. But then the $(C_A, C_B)$ min-cut separates $r$ from $r'$, and hence $T_{r} \subseteq C_A$ ad $T_{r'} \subseteq C_B$. It follows that $T_{r} \cap T_{r'} = \emptyset$.
    \end{proof}

    We are now ready to prove the correctness of our algorithm.  Suppose there exists a terminal $r^* \in R$ whose minimum isolating cut $(C_{r^*}, V \setminus C_{r^*})$ satisfies $E(C_{r^*}, V \setminus C_{r^*}) \leq \tau$. Using \Cref{claim: no crossing}, it follows that we must have $C_{r^*} \subseteq T_{r^*}$ (Recall that $T_{r^*}$ is defined as the set of vertices reachable from $r^*$ in the graph obtained after deleting the minimum cuts for every bi-partition in $\CC$). 
    Furthermore, by the definition of isolating cut, we know that $C_{r^*} \cap R = \{r^*\}$, therefore $C_{r^*} \subseteq T_{r^*}'$  
    Hence, the minimum $r^*$-$s_{r^*}$ cut in $G_{r^*}$ is of size at most $\tau$. Also by the definition of $R'$, we must have $r^* \in R'$.
    Since our algorithm finds an $r$-$s_r$ min-cut for each $r \in R'$, and we have $r^* \in R'$, the output cut must be a minimum isolating cut for some $r \in R$ of size at most $\tau$.

    \paragraph{Query complexity:} For each bi-partition $(A,B) \in \mathcal{C}$, we construct the auxiliary graph $H$ and compute an $s_{\text{source}}$-$s_{\text{sink}}$ maximum flow. For each edge $(s_{\text{source}},a)$, $a \in A$ with capacity $\tau + 1$, we replace it with $\tau + 1$ parallel edges each with capacity $1$. We then sub-divide each of these edges by adding an additional vertex. For each edge $(b,s_{\text{sink}})$ where $b \in B$, we do the same. Then the resulting graph is simple, has at most $O(n + |R|\tau) = O(n)$ nodes, and every edge has unit capacity. By~\Cref{thm:flow}, we can compute an $s_{\text{source}}$-$s_{\text{sink}}$ max-flow in this graph using $\Otil(n^{5/3})$ queries.
    
    The (restriction to $G$ of the) min-cut $(C_A,C_B)$ can be obtained as follows. Once we obtain a maximum flow $f_{A,B}$ we let $C_A$ be the set of vertices reachable from $s$ in the residual graph $G_{f_{A,B}}$, and let $C_B = V(G) \setminus C_A$. We can obtain $C_A$ by applying Corollary \ref{corollary: BFS tree finding} to find a BFS tree with root node $s_{\text{source}}$ in $G_{f_{A,B}}$ with $\Otil(n)$ queries. Once we know $(C_A,C_B)$ for each bi-partition $(A,B) \in \CC$, we can find the sets $T_{r}$ for each $r \in R$. We can then identify the set $R'$. For each $r \in R'$ we construct the graph $G_r$, and find an $r$-$s_r$ closest min-cut in $G_r$. To find this min-cut, we obtain a max-flow and find the set of vertices reachable from $r$. The next lemma shows that we can do this without using too many queries. 
    
    \begin{lemma}
    \label{lemma: local flow}
        We can compute $r$-$s_r$ max flow (min cut) in $G_r$ using $\Otil(|T_{r}'|^{5/3})$ queries.
    \end{lemma}
    \begin{proof}
    All flow goes from $r$ to $s_r$ can be divided into two cases
    \begin{enumerate}
        \item Directly goes from $r$ to $s_r$ through edge $(r, s_r)$
        \item First go through an edge $(r, u)$, where $u \in T_{r}' \setminus \{r\}$, and then goes from $u$ to $s_r$ 
    \end{enumerate}

    It's relatively easy to compute the flow in the first case by calculating the weight of $(r, s_r)$, which is $|E(r, V\setminus T_{r})|$ and can be determined with $O(1)$ queries. Notice that in the residual graph, all flows are the second case. The flow size of the second case is at most $|E(r, T_{r}' \setminus \{r\})| \le |T_{r}'| - 1$, since each flow must use an edge between $r$ to $C_r \setminus \{r\}$. Then by Theorem \ref{thm:flow}, we can compute it in $\Otil(|T_{r}' + 1|^{5/3}) = \Otil (|T_{r}|^{5/3})$ queries.
    \end{proof}

    By \Cref{lemma: separation,lemma: local flow}, the total query complexity for finding $r$-$s_r$ min-cuts for every $r \in R'$ is bounded by $\sum_{r \in R'} \Otil(|T_{r}|^{5/3}) = \Otil (n^{5/3})$. The final equality holds since the sets $T_{r}$, $r \in R'$, are disjoint from each other, ensuring $\sum_{r} |T_{r}| \le n$. This concludes the proof of \Cref{theorem: isolating cut}.
    \end{proof}

\subsection{Unbalanced Case}
In this section, we show how to find a cut of size at most $\tau$ when this cut is \emph{unbalanced} with respect to a given terminal set $R \subseteq V$ that is $\tau$-separated.
\label{subsection: unbalanced case}
\begin{definition} [Unbalanced/Balanced Cut]
    For any set of vertices $R \subseteq V$, a cut $C = (C_1, C_2 = V \setminus C_1)$ with size at most $\tau$, and a parameter $\phi \ge \poly(\frac{1}{\log n})$, we say that $C$ is $\phi$-unbalanced for $R$ if $\min \{|C_1\cap R|, |C_2\cap R|\}\le (\frac{1}{\phi})^3 + \frac{1}{\phi}$, otherwise we say that $C$ is $\phi$-balanced for $R$. 
\end{definition}

The goal of this subsection is to prove the following theorem.

\begin{lemma}
\label{theorem: unbalanced case}
    Let $R \subseteq V$ be  $\tau$-separated and $|R| \leq \Otil(\frac{n}{\tau})$. If there is a cut $(C, V \setminus C)$ of size at most $\tau$ that is  $\phi$-unbalanced for $R$ for some $\phi \geq \poly(\frac{1}{\log n})$, then we can find a cut with size at most $\tau$ in $\Otil(n^5/3)$ queries.
\end{lemma}

We need the following tool called \emph{splitter} for essentially  deterministic subsampling vertices.
\begin{lemma} [Theorem 4.3 from \cite{li2020deterministic}]
\label{lemma: unbalanced-case derandomize}
    For every positive integer $n$ and $k < n$, there is a deterministic algorithm that constructs a family $\cal{F}$ of subsets of $[n]$ such that, for each subset $S \subseteq [n]$ of size at most $k$, there exists a set $S' \in \cal{F}$ with $|S\cap S'| = 1$. The family $\cal{F}$ has size $k^{O(1)} \log n$ and contains only sets of size at least $2$. 
\end{lemma}

\begin{proof}[Proof of \Cref{theorem: unbalanced case}]
    Let $\mathcal{F}$ be the family 
    of subsets of $R$ from \Cref{lemma: unbalanced-case derandomize} such that, for any unbalanced cut $C$, there is a $R^* \in \mathcal{F}$ with $|R^* \cap C| = 1$. 
    Note that $|\mathcal{F}|=[(\frac{1}{\phi})^3 + \frac{1}{\phi}]^{O(1)} \log |R| = \Otil(1)$. 

    Then we run the algorithm of Theorem \ref{theorem: isolating cut} for each set in $\mathcal{F}$. Since $R^* \in \mathcal{F}$, the algorithm will successfully conclude that the isolating cut for $R^*$ is of size at most $\tau$. The total number of queries is $|\mathcal{F}|\cdot \Otil(n^{5/3}) = \Otil(n^{5/3})$.
\end{proof}

\subsection{Balanced Case}

We say that $R$ is \emph{$\phi$-balanced} if every cut $(C, V \setminus C)$ of size $|\partial C| \leq \tau$ is $\phi$-balanced for $R$. In the previous section, if $R$ is $\tau$-separated, but not $\phi$-balanced, then \Cref{theorem: unbalanced case} will find a cut of size at most $\tau$ for us. In this section, we handle the case when $R$ is $\tau$-separated and $\phi$-balanced using the following lemma.

\label{subsection: balanced case}
\begin{lemma}
\label{theorem: balanced case}
    Suppose that $R$ is $\tau$-separated and $\phi$-balanced. Then, we can make $\Otil(n^{5/3})$ cut queries and either 
    \begin{itemize}
        \item find a set $\wtilde{R}$ such that $|\wtilde{R}| \le O(\phi|R|\log^6 n) + \frac{|R|}{\log n}$ and $\wtilde{R}$ is $\tau$-separated, or
        \item find a cut in $G$ with cut-size at most $\tau$.        
    \end{itemize}
\end{lemma}

Towards proving this key lemma, we begin by introducing the notion of expanders and expander decomposition that we need.
\begin{definition} [$(\phi, R)$-expander] 
Given a graph $G$ and a terminal set $R \subseteq V$, $G$ is a $(\phi, R)$-expander if for every cut $(S, V \setminus S)$, we have

\[\Phi(S) = \frac{|\partial S|}{\min \{|R \cap S|,  |R \setminus S| \}} \ge \phi (\tau+1)\]
\end{definition}

We work with graphs with a slightly weaker notion of expansion, which we call $(\phi,R)$-almost expanders.

\begin{definition}
Given a  graph $G = (V,E)$ and a terminal set $R \subseteq V$, $G$ is a $(\phi, R)$ almost-expander with core $R'$ if $R' \subseteq R$, $|R'| \geq |R|(1 - \frac{1}{\log n})$, and for any cut $(S, V \setminus S)$ in $G$, we have
\[\Phi(S) = \frac{|\partial S|}{\min \{|R' \cap S|,  |R' \setminus S|\}} \ge \phi (\tau+1)\]

\end{definition}

Our next and crucial step in this section is  to obtain a decomposition into $(\phi, R)$-almost expanders. 
We defer this proof to \Cref{appendix: expander decomposition} since this proof is mostly standard.

At a high level, the algorithm combines an implementation of the cut-matching game~\cite{khandekar2007cut} together with the expander pruning technique~\cite{saranurak2019expander}. The cut player trivially needs zero queries since the auxiliary graph will be explicitly maintained. The matching player applies our algorithm for $s$-$t$ maximum flow. 
This takes $\Otil(n^{5/3})$ queries. Similar to the flow instance in \Cref{theorem: isolating cut}, we send flow between terminal set $R$ of size $\Otil(\frac{n}{\tau})$ and each terminal sends at most $\tau + 1$ units of flow. So the maximum flow is at most $\Otil(n)$ and, hence, our flow algorithm from \Cref{thm:flow} uses $\Otil(n^{5/3})$ queries. The formal statement is summarized below:

\begin{lemma}\label{claim:expander}
Given a graph $G = (V,E)$ and a terminal set $R \subseteq V$ with $|R| \leq \Otil(\frac{n}{\tau})$ and a parameter $\phi = \poly(\frac{1}{\log n})$, one can find a partition of the vertex set $V$ into subsets $V_1, V_2, \ldots, V_k$, where $R_i = V_i \cap R$ for each $i \in [k]$, and further obtain sets $R_1' \subseteq R_1, R_2' \subseteq R_2, \ldots, R_k' \subseteq R_k$ using $\Otil(n^{5/3})$ queries such that

\begin{enumerate}
\item Each $G[V_i]$ is a $(\phi, R_i)$ almost expander with core $R_i'$.
\item The number of crossing edges $\sum_{i\ne j} E(V_i, V_j) = O(\phi |R| (\tau+1)\log^6 n)$. 

\end{enumerate}

\end{lemma}

After obtaining this decomposition, we classify each part as follows.
\begin{definition}
    For each set $V_i$, $i \in [k]$, we say that $V_i$ is
    \begin{enumerate}
        \item  empty, if $R_i' = \emptyset$
        \item small, if $1 \le |R_i'| \le (\frac{1}{\phi})^2$
        \item large, if $|R_i'| > (\frac{1}{\phi})^2$
    \end{enumerate}
\end{definition}

Now, we construct $\wtilde{R}$ as follows:
\begin{enumerate}
    \item Include every vertex of $R_i \setminus R_i'$ in $\wtilde{R}$, for each $i \in [k]$.
    \item For each $V_i$, $i \in [k]$,
    \begin{itemize}
        \item if it is empty, do nothing;
        \item if it is small, add an arbitrary vertex of $R_i'$ to $\wtilde{R}$;
        \item else if it is large, add $1+\frac{1}{\phi}$ arbitrary vertices of $R_i'$ to $\wtilde{R}$.
    \end{itemize}

\end{enumerate}

\begin{lemma}\label{lem:still separate}
Suppose that $R$ is $\tau$-separated and $\phi$-balanced. 
If $|\partial V_i| \ge \tau+1$ for each $V_i$, then $|\wtilde{R}| \le O({\phi}|R|\log^6 n) + \frac{|R|}{\log n}$ and $\wtilde{R}$ is $\tau$-separated. 
\end{lemma}
\begin{proof}
    Since $|\partial V_i| \ge \tau+1$ and $\sum_{i\ne j}E(V_i, V_j) = O(\phi |R| (\tau+1)\log^6 n)$, the number of subsets is bounded by $k \le O(\phi |R|\log^6 n)$. Therefore, the total number of vertices in $\wtilde{R} \cap R_i'$ across all small $V_i$ is at most $O(\phi|R| \log^6 n)$. 
    
    For each large $V_i$, observe that we pick only an $O(\phi)$ fraction of $R_i'$ to add to $\wtilde{R}$. Additionally, for each $i$, we have $|R_i \setminus R_i'| \leq \frac{|R_i|}{\log n}$. Combining these results, we have 
    \[|\wtilde{R}| \le \frac{|R|}{\log n} + O(\phi |R|\log^6 n)\]
    as desired. Now we show that $\wtilde{R}$ must be $\tau$-separated.

   Assume, for the sake of contradiction, that $\wtilde{R}$ is not $\tau$-separated. This implies the existence of a cut $C = (C_1, C_2 = V(G) \setminus C_1)$ with cut-size $\partial C  \le \tau$ such that $\wtilde{R} \subseteq C_1$. We will demonstrate that this cut $C$ must be $\phi$-unbalanced for $R$,  contradicting our assumption that $R$ is $\phi$-balanced.

   Before we proceed further, we will show that, in this case, $C_2$ must contain the ``smaller'' side of every large $V_i$. The next lemma will clarify this point.

   \begin{claim}\label{lemma:expanderunbalance}
   If $C_2 \cap \wtilde{R} = \emptyset$, then for any large $V_i$, $i \in [k]$, we must have $|R_i' \cap C_2| \leq |R_i' \cap C_1|$.
   \end{claim}

   \begin{proof}
   Since $V_i$ is a $(\phi, R_i)$ almost-expander with core $R_i'$ and $(C_1, C_2)$ is a cut with cut-size at most $\tau$, we must have $\min\{|R_i' \cap C_1|, |R_i' \cap C_2|\} \leq \frac{\tau}{\phi(\tau + 1)} \leq \frac{1}{\phi}$.
   
   Thus if $|R_i' \cap C_1| \leq |R_i' \cap C_2|$, then $|R_i' \cap C_1| \leq \frac{1}{\phi}$. Recall that we picked $1 + \frac{1}{\phi}$ vertices from $R_i'$ and added them to $\wtilde{R}$. It then follows that at least one vertex of $R_i' \cap C_2$ was added to $\wtilde{R}$, and hence $\wtilde{R} \cap C_2 \neq \emptyset$, which is a contradiction. 
   \end{proof}

    Now, we show that $C$ is $\phi$-unbalanced for $R$.
    First, observe that 
    $\sum_{\text{large } V_i}|R_i' \cap C_2| \le \frac{1}{\phi}$; otherwise, 
    \[|\partial C| \ge \sum_{\text{large } V_i}\phi (\tau+1)  \min\{|R_i' \cap C_1|, |R_i' \cap C_2|\} = \phi (\tau+1)\sum_{\text{large } V_i} |R_i' \cap C_2| \ge \tau+1.\]
    where equality holds by \Cref{lemma:expanderunbalance}.
    
    Second, the number of small $V_i$ satisfying $R_i' \cap C_2 \neq \emptyset$ must be at most $\frac{1}{\phi}$; otherwise,
    \[|\partial C| \ge \sum_{\text{small } V_i} \phi (\tau+1) \min\{|R_i' \cap C_1|, |R_i' \cap C_2|\} \ge \sum_{\text{small } V_i, R_i' \cap C_2 \ne \emptyset} \phi (\tau+1) \ge \tau+1\]
    where we use the fact that $|R'_i\cap C_1| \ge 1$, otherwise we must pick a vertex in $R'_i\cap C_2$ into $\tilde{R}$, which is a contradiction.
    
    In total, we have
    \[|R \cap C_2| = \sum_{i\in [k]} |R_i' \cap C_2| \le \sum_{\text{large } V_i}|R_i' \cap C_2| + \sum_{\text{small } V_i, |R_i' \cap C_2| \ne \emptyset} |R_i' \cap C_2| \le \frac{1}{\phi} + (\frac{1}{\phi})^3\]
    where the equality holds because $(R_i \setminus R_i') \cap C_2 = \emptyset$ for any $i \in [k]$; otherwise, we have $C_2 \cap \wtilde{R} \neq \emptyset$.
    This implies that $C$ is $\phi$-unbalanced for $R$, contradicting that $R$ is $\phi$-balanced.
\end{proof}

We are now ready to prove \Cref{theorem: balanced case}.

\begin{proof}[Proof of \Cref{theorem: balanced case}]
First, we compute a $(\phi, R)$ almost-expander decomposition using \Cref{claim:expander}, which results in the sets     $V_1, V_2, \ldots, V_k$, $R_1, R_2, \ldots, R_k$ and the cores $R_i' \subseteq R_i$ for each $i \in [k]$.

Next, we make $\Otil(n)$ queries to determine if there exists a set $V_i$ such that $|\partial{V_i}| \leq \tau$. If such a set exists, we return this cut and terminate.

If no such set exists, we construct $\wtilde{R} \subseteq R$ as described previously. In this case, since $|\partial{V_i}| \geq \tau + 1$ for each $i \in [k]$,~\Cref{lem:still separate} guarantees that $|\wtilde{R}| \leq O(\frac{1}{\phi}|R|\log^6 n) + \frac{|R|}{\log n}$ and that $\wtilde{R}$ is $\tau$-separated. 
Then, we can return $\wtilde{R}$ as desired.
\end{proof}

\Cref{theorem: main theorem} now follows by combining~\Cref{theorem: unbalanced case} and~\Cref{theorem: balanced case}.

\begin{proof}[Proof of~\Cref{theorem: main theorem}]
We begin by computing a dominating set $R$ using~\Cref{theorem: find a dominating set}. 

Now let $\phi = 1/\log^{10} n$. 
Using both  \Cref{theorem: unbalanced case,theorem: balanced case} with the set $R$ as the terminal set, we proceed in one of two ways: if we find a cut with size at most $\tau$ in $\Otil(n^{5/3})$ queries, then we are done. Otherwise, if we find a subset $R' \subseteq R$ such that $|R'| \le \frac{1}{2}|R|$ and $R'$ is still  $\tau$-separated. We replace $R$ by $R'$ and recursively continue the algorithm till we either find a cut of size at most $\tau$, or confirm that no such cut exists.

Since at each step $|R'| \le \frac{1}{2}|R|$ by the choice of $\phi$, so the recursion can proceed for at most $\mathcal{O}(\log |R|) = \Otil(1)$ iterations. This process guarantees that we will eventually find a cut of size at most $\tau$ or declare that no such cut exists.
Finally, note that the total query complexity is $\Otil(n^{5/3})$ since we make at most $\Otil(1)$ calls to the subroutines in~\Cref{theorem: unbalanced case,theorem: balanced case}.
\end{proof}

\section{Conclusion}
We show the first subquadratic deterministic algorithms for the $s$-$t$ minimum cut and global minimum cut problems in simple graphs, both using $\Otil(n^{5/3})$ cut queries. 
Nevertheless, there remains a considerable gap between our results and the current lower bound of $\Omega(n)$ \cite{hajnal1988communication}. 
Improving the lower bound to $\omega(n)$ is very interesting as it would separate deterministic and randomized cut query complexity for the global min-cut problem. 

For the upper bound side, algorithms using $\Otil(n)$ queries would be exciting. This algorithm must be very different from ours because we explicitly compute a maximum $s$-$t$ flow and the flow of value $\nu$ may have representation size as large as $\Omega(n \sqrt{\nu})=\Omega(n^{1.5})$ even on simple unweighted graphs. 
In fact, it is known that, given any \emph{simple} unweighted graph $G$ where the maximum $s$-$t$ flow value is $\nu$, there exists a subgraph $H\subseteq G$ with $O(n \sqrt{\nu})$ edges such that the maximum $s$-$t$ flow in $H$ has the same value as the one in $G$  and this bound is tight (see, e.g., \cite{karger1998finding,rubinstein2018computing}). Thus, it is interesting whether there is a (near-optimal) algorithm for explicitly computing maximum $s$-$t$ flow in a simple unweighted graph using $\Otil(n \sqrt{\nu})$ cut queries. Through our framework, this would immediately imply a global min-cut algorithm using $\Otil(n^{1.5})$ cut queries, reaching the barrier of this approach.

It is interesting to generalize our result to weighted graphs or just non-simple unweighted graphs. Our technique does not work because \Cref{lemma:dominating set is t-separated} is specific to simple unweighted graphs.

\bibliographystyle{alpha}
\bibliography{ref}

\appendix

\section{Computing Expander Decomposition: Proof of~\Cref{claim:expander}}
\label{appendix: expander decomposition}

\paragraph{Definitions.}
Given a graph $H = (V(H),E(H))$, define the \emph{sparsity} of a cut $(S, V(H) \setminus S)$ as $\frac{|E(H, V(H) \setminus S)|}{\min\{|S|, |V(H) \setminus S|\}}$. Further, given a terminal set $R$, we define the \emph{sparsity with respect to $R$} of a cut $(S, V(H) \setminus S)$ as $\frac{|E(H, V(H) \setminus S)|}{\min\{|S \cap R|, |(V(H) \setminus S) \cap R|\}}$. We say that the set $R$ is $\phi$-expanding in $H$ if there is no cut $(S, V(H) \setminus S)$ in $H$ which has sparsity at most $\phi$ with respect to $R$. Further, $H$ is $\phi$-expanding if $V(H)$ is $\phi$-expanding in $H$. We say that a cut $(S, V(H) \setminus S)$ is $b$-balanced with respect to $R$ if both $S \cap R$ and $(V(H) \setminus S) \cap R$ have size at least $b|R|$. For a set of vertices $S \subseteq V(H)$, we define the volume of $S$, $\vol_H(S)$, as the sum of the degrees of the vertices in $S$ in $H$. The conductance of a cut $(S, V \setminus S)$ in $H$ is defined as $\frac{E(S, V \setminus S)}{\min\{\vol_H(S), \vol_H(V \setminus S)
\}}$. The conductance of $H$ is the minimum conductance across all such cuts. We drop the subscripts when the graph $H$ is clear from the context.

The following is the key subroutine we need for finding a decomposition into $(\phi, R)$ almost-expanders.

\begin{lemma}[One-step of Expander Decomposition]
\label{thm:onestep}
 Given a graph $G=(V,E)$, terminal set $R\subseteq V$ and parameters $\phi = \poly(\frac{1}{\log n})$ and $\tau \geq 1$ such that $|R| \leq \Otil(\frac{n}{\tau})$, using $\Otil(n^{\frac{5}{3}})$ queries we can either 
\begin{itemize}
\item return a cut $(S,\overline{S})$ that is $(\frac{1}{\log^5 n})$-balanced and $(\phi \tau)$-sparse with respect to $R$.
\item return a set $R'\subseteq R$ of size $|R'|\ge(1-\frac{1}{\log n})|R|$ such that $R'$ is $\frac{\phi \tau}{\log^5 n}$-expanding in $G$.
\end{itemize}
\end{lemma}

Before showing this crucial lemma, let us first show why it implies~\Cref{claim:expander}.

\begin{proof}[Proof of~\Cref{claim:expander}]
We start with the input graph $G$ and run the algorithm of~\Cref{thm:onestep}. If~\Cref{thm:onestep} returns a cut $(S, \overline{S})$, then we recursively apply~\Cref{thm:onestep} on $G[S]$ and $G[\overline{S}]$. On the other hand, if~\Cref{thm:onestep} returns a set $R'$, then $G$ is a $(\phi,R)$ almost expander with core $R'$. The recursion depth for the first case is at most $\log^5 n \log |R|$, and every leaf of the recursion tree, where the second case occurs, together give us the decomposition into $(\phi,R)$ almost expanders.

The argument about the number of crossing edges is a standard charging argument: Every time we obtain a cut $(S, \overline{S})$, that is $\frac{1}{\log^5 n}$ balanced and $\phi (\tau+1)$ sparse with respect to $R$, we charge the number of cut edges to the terminals on the side of the cut with smaller number of terminals. Since each cut is $\phi (\tau+1)$ sparse, each terminal $r \in R$ receives a charge of at most $\phi (\tau+1)$. Since every cut we find is $\frac{1}{\log^5 n}$ balanced, the recursion depth is at most $\log^5 n \log |R|$, and hence a terminal is charged at most $\log^5 n \log|R| \leq \log^6 n$ times. It follows that the number of cut edges is at most $\OO(\phi|R|(\tau+1)\log^6 n)$.    
\end{proof}

We now proceed to prove~\Cref{thm:onestep}. The key ingredient is the cut-matching game of \cite{khandekar2007cut}. Here, we state the cut-matching game where in each round the matching player returns a perfect \emph{$b$-matching} instead of a perfect matching. A perfect $b$-matching is a set of edges such that each vertex $v$ has exactly $b$ edges in the set containing $v$. A perfect $b$-matching between two sets $A$ and $B$ is a perfect $b$-matching in which every edge is between a vertex of $A$ and a vertex of $B$. We now describe the cut-matching framework.

\begin{itemize}

\item Start with an empty graph $X$ with $n$ vertices.

\item In round $i$, the cut player chooses a bisection $(A,B)$ of $V(X)$. 

\item Then, the matching player returns an arbitrary perfect $b$-matching $M_{i}$ between $A$ and $B$, and the edges in $M_i$ are added to $E(X)$.

\end{itemize}

The goal of the cut player is to guarantee that, after a few rounds, the graph $X$ (whose edge set is the union $\bigcup_{i}M_{i}$) has sparsity at least $b$, i.e., for every $S\subseteq V(X)$, we have
\[
E_{X}(S,V(X)\setminus S)\ge b\cdot\min\{|S|,|V(X)\setminus S|\}.
\]
\begin{theorem}[\cite{khandekar2007cut}]\label{thm:cutplayer}
In the cut-matching game where the matching player always returns a perfect $b$-matching in each round, there exists an algorithm for the cut player that, after $r=O(\log n)$ rounds, guarantees that $X$ has sparsity at least $b$. 
\end{theorem}

We remark that the running time of~\Cref{thm:cutplayer} is exponential, but since we will know the entire graph $X$ along with its edge set $E(X)$ whenever we invoke this theorem, we spend zero cut queries to run this algorithm, and hence it does not affect the query complexity of our algorithm. However, if we insist on a polynomial time implementation, we can still guarantee that $\bigcup_{i=1}^{r}M_{i}$ has sparsity at least $b/O(\sqrt{\log n})$ using~\cite{arora2009expander}. For the rest of the section, for simplicity, we work with the exponential time cut player as in~\Cref{thm:cutplayer}.

Before we prove~\Cref{thm:onestep}, we need an implementation for the matching player. Towards this, let us define the notion of an \emph{almost perfect $b$-matching}.

\begin{definition}
Given sets $A,B$, parameters $\beta,b$ and a graph $H$, a set of edges $F$ is a $\beta$-almost perfect $A$-$B$ $b$-matching, if there exists a set of $\beta b$ edges $F'$, such that $F \cup F'$ is a perfect $b$-matching between $A$ and $B$.

\end{definition}

The matching player, in each round, either finds a sufficiently balanced and sparse cut with respect to the terminal set $R$, or returns an embedding of an almost perfect $b$-matching. In every round, if we find a balanced and sparse cut, we return this cut. Else, we use the almost perfect matching, together with a small number of fake edges, to obtain a perfect matching and make progress for the cut-matching game.

Before stating this theorem, we define the notion of a flow embedding. Given a set of edges $M$, we say that  $M \preceq^{\text{flow}} \frac{1}{\phi}G$ to mean that there is a multicommodity flow in $G$ that can simultaneously send one unit of flow across each edge in $M$ with congestion $\frac{1}{\phi}$.

\begin{theorem} [Matching Player]
\label{theorem: matching player}
    Suppose we are given a graph $G = (V,E)$, integers $\tau, \beta \geq 1$, a terminal set $R \subseteq V$ such that $|R| \leq \Otil(\frac{n}{\tau})$ and a parameter $\phi = \poly(\frac{1}{\log n})$. Given a bisection $(A,B)$ of $R$, s.t. $|A| = |B| = \frac{|R|}{2}$ 
 there is an algorithm that either
    \begin{itemize}
        \item [(1)]  finds a cut $(S^*, V \setminus S^*)$ with $\min\{|S^*\cap A|, |(V \setminus S^*) \cap B|\} \ge \beta$, such that $(S^*, V \setminus S^*)$ is $\phi(\tau+1)$ sparse with respect to $R$.
        \item [(2)] obtains a flow embedding of a $\beta$-almost perfect $A$-$B$ $(\tau+1)$-matching $M$ such that
        $M \preceq^{\text{flow}} \frac{1}{\phi}G$
    \end{itemize}
    using $\Otil(n^{\frac{5}{3}})$ queries.
    
\end{theorem}
\begin{proof}
    Add two vertices $s_{source}$ and $s_{sink}$ to the original graph. Connect $s_{source}$ to all vertices in $A$ with capacity $\tau+1$, and $s_{sink}$ to all vertices in $B$ with capacity $\tau+1$, and we give all edges in $G$ with capacity $\frac{1}{\phi}$.

    Then since $|R| \leq \Otil(\frac{n}{\tau})$, the max-flow from $s$ to $t$ is at most $\Otil(n)$. We can then compute an $s_{source}$-$s_{sink}$ max flow in $\Otil(n^{\frac{5}{3}})$ queries by Theorem \ref{thm:flow}\footnote{As in previous instances, we can add parallel edges and subdivide the edges incident to the source and sink to ensure that all edges have unit capacity}. If the maximum flow is of size $\geq (|A| - \beta)(\tau+1)$, then using the set of flow paths from the flow algorithm, we obtain a $\beta$-almost perfect $A$-$B$ $(\tau+1)$ matching $M$ such that $M \prec^{\text{flow}} \frac{1}{\phi} G$. Otherwise, after finding the max flow $f^*$ from $s_{source}$ to $s_{sink}$, we run BFS (see Corollary \ref{corollary: BFS tree finding}) from $s_{source}$ in the residual graph $G_{f^*}$ and find all vertices that can be reached by $s_{source}$, denoted as $S^*$. Then $|S^* \cap A| \ge \beta$, since there are at least $\beta$ vertices in $A$ not sending all the demands. Similarly we can show that $|\overline{S^*} \cap B| \ge \beta$, where $\overline{S^*} = V \setminus S^*$. Now we show that $(S^*, \overline{S^*})$ is $\phi(\tau+1)$ sparse with respect to $R$. Without loss of generality, let us assume that $|S^* \cap R| \leq |\overline{S^*} \cap R|$. Recall that the capacity of every edge in $E$ is set to be $\frac{1}{\phi}$. Since the flow is feasible, and every edge of the min-cut is saturated in the maximum flow, we must have $\frac{|\partial S^*|}{\phi} \leq |S \cap R|(\tau+1)$, which means that $|\partial S^*| \leq \phi|S \cap R|(\tau+1)$. It follows that $(S^*, \overline{S^*})$ is a $\phi(\tau + 1)$-sparse cut with respect to $R$.
\end{proof}

\begin{proof}[Proof of~\Cref{thm:onestep}]
We run the cut player algorithm in~\Cref{thm:cutplayer}, and the matching player algorithm in~\Cref{theorem: matching player} with $\beta  = \frac{|R|}{\log^5 n}$. If in any round the matching player returns a cut that is  $\phi(\tau+1)$ sparse with respect to $R$ with both sides having at least $\beta$ terminals, we simply return this cut as this cut satisfies the requirement of~\Cref{thm:onestep}. Otherwise, in every round, the matching player returns an $\beta$-almost perfect $(\tau+1)$-matching $F$, so that there exists a set $F'$ with $|F'| \leq \beta (\tau+1)$ such that $F \cup F'$ is a perfect $b$-matching and we add the edges in $F \cup F'$ to $X$ to continue the cut-matching game. Mark the edges $F'$ as \emph{fake edges}.

After $r = \mathcal{O}(\log n)$ rounds of the cut-matching game, if the matching player did not return a balanced and sparse cut in any round,~\Cref{thm:cutplayer} guarantees that $X$ has sparsity at least $b = \tau+1$. Note that the total number of fake edges is at most $\beta (\tau+1)r$. Also, since each matching flow embeds in $\frac{1}{\phi} G$, and there are at most $O(\log n)$ rounds, we must have $X \preceq^{\text{flow}} \frac{\OO(\log n)}{\phi}G$.

Consider any cut $(S, V(X) \setminus S)$ with $|S| \leq |V(X) \setminus S|$ in $X$. Since the sparsity is at least $(\tau+1)$, it follows that $|E(S, V(X) \setminus S)| \geq (\tau+1)|S|$. Note that in $X$, every vertex has degree $(\tau+1)r = \mathcal{O}((\tau+1)\log n)$.

This means that $\vol(S) \leq \mathcal{O}(|S| (\tau+1) \log n)$, and hence the conductance of every cut $(S, V(X) \setminus S)$ is at least $\Omega(\frac{1}{\log n})$.

We will now show that there exists a set $R'$ such that $R' \subseteq R$, $|R'| \geq |R|(1 - \frac{1}{\log n})$, and $R'$ is $\phi (\tau+1)$ expanding in $X$. The construction of $R'$ is simple. On a high level, our goal is to try to certify the expansion of $X$ after removing the fake edges $F'$. Therefore, a natural step is to apply expander pruning to the graph $X$ after deleting the set of fake edges $F'$.

\begin{theorem}[\cite{saranurak2019expander}, Theorem 1.3]
Given a graph $H$ in which every cut has conductance at least $\phi'$, and a set of deletion edges $F'$, there exists a prune set $P \subseteq V(H)$ with $\vol_H(P) \leq \frac{8}{\phi'} |F'|$, such that in the induced subgraph on $V(H) \setminus P$ after removing $F'$, in $H[V(H) \setminus P] \setminus F'$, every cut $(S, (V(H) \setminus P) \setminus S)$ has conductance at least $\frac{\phi'}{6}$. 
\end{theorem}

Apply expander pruning to graph $X$, with the set of fake edges $F'$ as the deletion set, and the conductance parameter $\phi_X = \Omega(\frac{1}{\log n})$ to obtain a prune set $P$. Here we recall that the graph $X$ is explicitly known, so we spend zero queries to obtain this set $P$. Note that $\vol(P) \leq \frac{8|F'|}{\phi_X} \leq \OO(\beta(\tau+1)r \log n) = \OO(\beta (\tau+1)\log^2 n ) \leq \frac{|R|(\tau+1)}{\log^2 n}$ since $\beta = \frac{|R|}{\log^5 n}$. But each vertex has degree $r(\tau+1)$ in $H$, hence the size of $P$, $|P| \leq \frac{\vol(P)}{r(\tau+1)} \leq \OO(\beta \log n) \leq \frac{|R|}{\log^2 n}$. We define $R'$ as all the vertices $v \in R \setminus P$, which have a combined at most $\frac{\tau+1}{10}$ number of edges incident to $P$ and edges incident in $F'$. Note that $|F'| + \vol(P) \leq 2\vol(P) \leq \frac{2|R|(\tau+1)}{\log^2 n}$.

By Markov's inequality, the number of $v \in R$ for which there are more than $\frac{\tau + 1}{10}$ edges in $F$ or whose other endpoint is in $P$, must be at most $20\frac{|R|}{\log^2 n} \leq \frac{|R|}{2\log n}$. But this means $|R'| \geq |R|(1 - \frac{1}{2\log n}) - |P| \geq |R|( 1- \frac{1}{\log n})$ since $|P| \leq \frac{|R|}{\log^2 n}$.

Since each vertex $v \in R$ has degree $r(\tau+1)$, it follows that each vertex in $R'$ has degree at least $\frac{r}{10}(\tau+1)$ in $X[V(X) \setminus P] \setminus F'$. The next lemma proves the desired expansion property of $R'$.

\begin{lemma}
$R'$ is $\frac{\phi(\tau+1)}{\log^5 n}$ expanding in $G$.
\end{lemma}

\begin{proof}
Consider a cut $(S, V \setminus S)$ in $G$ with $|S \cap R'| \leq |(V \setminus S) \cap R'|$. Let $R_1 = S \cap (R \setminus P)$ and $R_2 = (V(X) \setminus S) \cap (R \setminus P)$. Since the induced subgraph on $V(X) \setminus P \setminus F$ is a $\Omega(\frac{1}{6\log n}) \geq \frac{1}{\log^2 n}$ conductance expander, there are at least $\frac{\min(\vol(R_1), \vol(R_2))}{6\log^2 n} \geq \frac{|S \cap R'|(\tau+1)}{60\log^2 n}$ edges in $E(X)$ between the sets $R_1$ and $R_2$, where the inequality follows since each vertex of $R'$ has degree at least $\frac{\tau + 1}{10}$ in $X[V(X) \setminus P] \setminus F'$. Recall that $X \preceq^{\text{flow}}O(\frac{\log n}{\phi}) G$. Thus the number of edges across $(S, V \setminus S)$ in $G$ is at least $\Omega\left (\frac{|S \cap R'|\phi(\tau+1)}{\log^3 n}\right ) \geq \frac{\phi(\tau + 1)|S \cap R'|}{\log^5 n}$.  
\end{proof}

\end{proof}

\end{document}